\documentclass[10pt, a4paper]{article}
\usepackage{graphicx,latexsym, amsmath,amsfonts}
\usepackage{amsthm}
\usepackage{enumerate}
\usepackage{epsfig}
\usepackage{verbatim}

\makeatletter \@addtoreset{equation}{section}

\begin{document}

\linespread{1.3}

\newcommand{\E}{\mathbb{E}}
\newcommand{\PP}{\mathbb{P}}
\newcommand{\RR}{\mathbb{R}}
\newcommand{\NN}{\mathbb{N}}

\newtheorem{theorem}{Theorem}[section]
\newtheorem{conjecture}{Conjecture}
\newtheorem{lemma}[theorem]{Lemma}
\newtheorem{coro}[theorem]{Corollary}
\newtheorem{defn}[theorem]{Definition}
\newtheorem{assp}[theorem]{Assumption}
\newtheorem{prop}[theorem]{Proposition}
\theoremstyle{remark}
\newtheorem{remark}[theorem]{Remark}
\newtheorem{example}[theorem]{Example}
\newtheorem{str}{Strategy}
\newtheorem{notation}[theorem]{Notation}

\newcommand\tq{{\scriptstyle{3\over 4 }\scriptstyle}}
\newcommand\qua{{\scriptstyle{1\over 4 }\scriptstyle}}
\newcommand\hf{{\textstyle{1\over 2 }\displaystyle}}
\newcommand\hhf{{\scriptstyle{1\over 2 }\scriptstyle}}

\newcommand{\eproof}{\indent\vrule height6pt width4pt depth1pt\hfil\par\medbreak}

\def \Ito {{It{\^o} }}
\def \Levy {{L{\'e}vy }}

\def\a{\alpha} \def\g{\gamma}
\def\e{\varepsilon} \def\z{\zeta} \def\y{\eta} \def\o{\theta}
\def\vo{\vartheta} \def\k{\kappa} \def\l{\lambda} \def\m{\mu} \def\n{\nu}
\def\x{\xi}  \def\r{\rho} \def\s{\sigma}
\def\p{\phi} \def\f{\varphi}   \def\w{\omega}
\def\q{\surd} \def\i{\bot} \def\h{\forall} \def\j{\emptyset}

\def\be{\beta} \def\de{\delta} \def\up{\upsilon} \def\eq{\equiv}
\def\ve{\vee} \def\we{\wedge}

\def\t{\tau}

\def\F{{\cal F}}
\def\T{\tau} \def\G{\Gamma}  \def\D{\Delta} \def\O{\Theta} \def\L{\Lambda}
\def\X{\Xi} \def\S{\Sigma} \def\W{\Omega}
\def\M{\partial} \def\N{\nabla} \def\Ex{\exists} \def\K{\times}
\def\V{\bigvee} \def\U{\bigwedge}

\def\1{\oslash} \def\2{\oplus} \def\3{\otimes} \def\4{\ominus}
\def\5{\circ} \def\6{\odot} \def\7{\backslash} \def\8{\infty}
\def\9{\bigcap} \def\0{\bigcup} \def\+{\pm} \def\-{\mp}
\def\<{\langle} \def\>{\rangle}

\def\lev{\left\vert} \def\rev{\right\vert}
\def\1{\mathbf{1}}

\def\tl{\tilde}
\def\trace{\hbox{\rm trace}}
\def\diag{\hbox{\rm diag}}
\def\for{\quad\hbox{for }}
\def\refer{\hangindent=0.3in\hangafter=1}

\newcommand{\Xk}{X_{t_{k}}}
\newcommand{\Xkk}{X_{t_{k+1}}}

\newcommand{\Yk}{Y_{t_{k}}}
\newcommand{\Ykk}{Y_{t_{k+1}}}

\newcommand\wD{\widehat{\D}}

\title
{ \bf A limit order book model for latency arbitrage.}

\author{Samuel N. Cohen  %
          \thanks{%
                  Mathematical Institute,
                  University of Oxford,
                  24-29 St Giles,
                  Oxford OX1 3LB, UK
                  (\texttt{samuel.cohen@maths.ox.ac.uk}).
                  }
\and
       Lukasz Szpruch%
        \thanks{%
                  Mathematical Institute,
                  University of Oxford,
                  24-29 St Giles,
                  Oxford OX1 3LB, UK
                  (\texttt{szpruch@maths.ox.ac.uk}).
                 }
       }

\date{}

\maketitle

\thispagestyle{empty}


\begin{abstract}
\textsf{\em  We consider a single security market based on a limit order book and two investors, with different speeds of trade execution. If the fast investor can front-run the slower investor, we show that this allows the fast trader to obtain risk free profits, but that these profits cannot be scaled. We derive the fast trader's optimal behaviour when she has only distributional knowledge of the slow trader's actions, with few restrictions on the possible prior distributions. We also consider the slower trader's response to the presence of a fast trader in a market, and the effects of the introduction of a `Tobin tax' on financial transactions. We show that such a tax can lead to the elimination of profits from front-running strategies. Consequently, a Tobin tax can both increase market efficiency and attract traders to a market.}


\medskip
\noindent \textsf{{\bf Key words: } \em limit order book, latency arbitrage, high-frequency trading, Tobin tax.}

\medskip
\noindent{\small\bf 2000 Mathematics Subject Classification: } 91B26, 91G99

\end{abstract}

\section{Introduction}
The study of investment decisions in the presence of a market based on a limit order book poses various challenges in the fields of mathematical finance and financial economics. Many phenomena which are classically overlooked, in particular the existence of a bid/ask spread and the price impact of trading, appear naturally in this context. These phenomena take on a particular importance when we consider decisions over a short time horizon, where the effects of the limit order book naturally outweigh long-term price fluctuations.

In this paper, we present a model in the spirit of
\cite{jarrow2010liquidity,schied2011some,roch2011resilient,gokay2011liquidity,gatheral2011exponential,alfonsi2009optimal}
that combines classical financial modelling with a limit order book.
Inspired by recent work of Jarrow and Protter \cite{jarrow2011dysfunctional}, 
we investigate whether technological developments, in particular the advent of `high-frequency trading', may lead to some forms of risk-free profit, commonly called \emph{latency arbitrage}. By linking this trading to its effects on the limit order book, we are able to give an economically justifiable explanation for these abnormal profits, including the prevention of unbounded profit.

We build a model of the interactions of a fast and a slow trader with the limit order book. By considering the fast trader's competitive advantage in terms of order execution, and by allowing her to predict the slow trader's actions, we can easily see that she can create a risk-free profit by using a `front-running' strategy. From this we can give precise estimates on the various profits that the fast and slow traders are able to make, in terms of the underlying limit order book. We also model their interaction at equilibrium and the consequent market efficiency, and the effects of imposing a tax on transactions (a Tobin tax).

\subsection{Instantaneous trading}

In order to focus on the effects of high-frequency trading, rather than on the other complexities of markets with limit order books, we shall focus our attention on trades occurring at a single moment. This allows us to keep our attention on the observable qualities of the limit order book (its depth, the number of stocks available at a price, and tightness, the cost of turning around a position), and to avoid building a model for its resilience (the rate at which new orders enter the book, c.f. Kyle \cite{kyle1985continuous}).

From a mathematical perspective, focussing our attention on a single moment allows us to easily consider more complex models for the instantaneous behaviour of the limit order book. In particular, the bid and ask prices we observe will naturally jump whenever trades occur, which causes no mathematical difficulties. Our fast-trader's strategies can also be modelled at a point, without assuming either right- or left-continuity. In some ways, this can be seen as similar to the model in \cite{jarrow2011dysfunctional}, where prices are assumed to be continuous, except at a countable number of jump times. In their model, the high-frequency trader only acts as a high-frequency trader at the jump times (albeit with right-continuous actions), and so the key details follow from the analysis of this countable set of stopping times. Our approach is to reduce this analysis to a single time, restoring the model of \cite{jarrow2011dysfunctional} is then primarily a question of notation.

On the other hand, by focussing on a single moment, we avoid consideration of the theory of optimal execution, as in \cite{alfonsi2009optimal, obizhaeva2005optimal,  almgren2003optimal, subramanian2001liquidity, huberman2005optimal, schoneborn2007liquidation,engle2006execution}. The majority of this theory revolves around the resilience of the limit order book, and how to exploit this resilience to minimise price impact. In the time scales considered for high-frequency trading, the resilience of the limit order book is a less important consideration, as the competitive advantage of the high-frequency trader disappears faster than the limit order book returns to equilibrium.

Our main attention will be on `predatory' trades made by a high-frequency trader using a front-running strategy, in such a way that they have no net position before or after the trades are complete. Such a sequence of trades is often called a `round-trip'. Studying round-trips is a common feature of many models of limit order books (eg \cite{alfonsi2009optimal, huberman2004price, gatheral2010transient}), and is also practically significant for high-frequency traders. For example, Kirilenko et al \cite{kirilenko2010flash} describe the `Hot Potato Effect' before the flash-crash of May 6, 2010, when in a 14 second period, high-frequency traders exchanged over 27000 contracts while only changing their net position by about 200 contracts. By studying these transactions, we can model our high-frequency trader as having a high temporal-risk-aversion, they wish to make a profit by exploiting their speed advantage, and do not wish to expose themselves to price risk through time.

\subsection{Foreknowledge and high-frequency trading}\label{sec:foreknowledge}

For clarity of exposition, suppose that there are traders Alice (fast) and Bob (slow).
Alice and Bob may both submit their orders effectively at the same time, and based on the
same information. However, Alice's superior speed (typically due to Alice's server being
located closer to the exchange) leads to Alice's order being executed prior to Bob's. 

Alice, our high-frequency trader, will wish to make profit by exploiting her higher trading speed. Her primary effect on Bob is that, as her orders are executed first, she will have an impact on the price at which Bob can trade. This effect is a form of slippage (or temporary price impact), and is also modelled in \cite{cvitanic2010high, alfonsi2009optimal}. While this gives Alice a comparative advantage, it does not allow her to obtain a risk-free profit at a single instant.

To allow Alice to make such a profit, we will suppose that Alice, as a high-frequency trader, is able to know the quantity Bob intends to trade trade instantaneously before Bob does so. In this setting, Alice is able to manipulate the limit order book in such a way as to make a profit from Bob.

If we assume that Bob is an algorithmic trader, then Bob's behaviour is predictable if the algorithm is known. Hence, if Alice has a good approximation of Bob's algorithm, Bob's trades are susceptible to the strategies that we analyse. From an example of Jarrow and Protter, \cite[Example 1]{jarrow2011dysfunctional}, by placing and cancelling small orders and watching Bob's reactions, Alice can `learn' Bob's algorithm for a very low cost, and can then front-run it.  From an industry perspective, Arnuk and Saluzzi \cite{arnuk2009latency} also assume that a high-frequency trader knows the trades coming in to the market, supporting our basic model.

As emphasized by Moallemi et al. \cite{moallemi2008strategic} the optimal execution strategies derived by Bertsimas and Lo \cite{bertsimas1998optimal} and others lead to this foreknowledge, as the optimal strategy for liquidating a position involves an equipartitioning through time, and is therefore predictable, and can be detected in the market. Moallemi et al \cite{moallemi2008strategic} then construct Nash equilibrium strategies with Bayesian updating, under certain assumptions on the form of optimal strategies (it must be Gaussian and linear in certain parameters). In Section \ref{sec:imperfectknowledge}, we extend our model to encompass one side of this decision making, as we determine Alice's behaviour when she only has distributional knowledge of Bob's actions.

Such a foreknowledge assumption is also implicit in the model of predatory trading considered in Brunnermeier and Pedersen \cite{brunnermeier2005predatory}. In their model, the predatory trader knows that the other party must liquidate their position, and trades accordingly. In contrast to \cite{brunnermeier2005predatory}, by modelling a limit order book, we see that predatory trading of the type we consider will not lead to overshooting of the price, in fact, that the limit order book before and after trades are completed is the same whether the high-frequency trader is present or not. This is fundamentally because our high-frequency trader cannot, in an instant, increase the \emph{bid} price by purchasing stocks, as only the ask price will be affected. Therefore, in our model, a single trader cannot create profitable `momentum', and can only front-run the slow trader's exogenously defined trades.

By making this foreknowledge assumption, and by considering only behaviour at an instant, our model differs markedly from Cvitani\'c and Kirilenko \cite{cvitanic2010high}. In \cite{cvitanic2010high}, the high-frequency trader makes profits by taking orders away from the front of the limit order book, rather than by exploiting the shape of the book. This implies that the high-frequency trader must take on some temporal risk, and does not make a risk-free profit.

Our work also serves as a partial justification for the approach of Jarrow and Protter \cite{jarrow2011dysfunctional}, where it is assumed that the high-frequency trader is able to obtain profits from using an optional integral. We show that this is indeed the case in our model (see Section \ref{sec:optionalintegration}).

We note at this point that we do not model the possibility that, if Bob is an algorithmic trader, Alice can trade so as to trigger Bob's actions. We shall generally allow Bob's trades to be determined exogenously (or possibly by giving Bob a demand function, for the purposes of determining his response to Alice). This assumption prevents Alice from being able to independently generate profit from Bob, which could lead to Alice obtaining unrealistically large profits from exploiting Bob's algorithmic nature (see Remark \ref{rem:triggeringlimits}.)

\subsection{Market equilibrium}

As we shall see, our model does lead to a high-frequency trader being able to make a risk-free profit. However, this profit is bounded, as Alice can only profit from adjusting the limit order book up to the quantity that Bob is planning to trade. In this way, we see that while we may violate the principle of `no free lunch with vanishing risk' (c.f. Delbaen and Schachermayer \cite{delbaen2006mathematics}), due to nonlinearities Alice's payoff, we sill still satisfy the principle of `no unbounded profit with bounded risk' (c.f. Karatzas and Kardaras \cite{karatzas2007numeraire}). In this sense, our market is still economically reasonable, even though Alice is able to make a risk-free profit.

On the other hand, we shall also examine the behaviour of Alice and Bob when Bob becomes aware of Alice (Section \ref{sec:equilibrium}). In this setting, we shall see that, even if Bob is unable to avoid the costs imposed by Alice, he will continue to trade, albeit in smaller quantities. Therefore, Alice's presence does not render the market unworkable, even though it does introduce a deadweight loss in the market.

This equilibrium is different to that considered in \cite{biais2000competing}, as we do not study the formation of the limit order book at equilibrium, but rather the equilibrium interactions of Alice and Bob given the limit order book. It is also not the same as those in \cite{biais2009liquidity, cvitanic2010high}, as we do not look at dynamic equilibrium behaviour, instead focussing on a single instant.

One method of preventing Alice from interacting with the market is to impose a `Tobin' tax on financial transactions. Such a tax is currently being proposed by the European Commission (see \cite{europa}). We shall see that such a tax will either completely prevent Alice from participating in the market, or will not alter her behaviour at all. However, we also show that there exists a small range of possible tax rates which will prevent Alice from participating without imposing larger costs on the rest of the market. That is, unlike in most microeconomic models of taxation, there is a small band of tax rates which may improve the market efficiency. Such a result largely agrees with the analysis of Tobin taxes in \cite{palley1999speculation, ehrenstein2005tobin, westerhoff2003heterogeneous}, however here for the setting of high-frequency trading. On the other hand, we shall see that outside this band, the tax does nothing to prevent Alice's predatory trading, and will impose further inefficiencies on the market.

Furthermore, we show that for an individual slow investor, the negative effects of imposing the tax can be completely outweighed by the prevention of high-frequency trading, making a market more attractive to trade. This suggests that the `flight of capital' historically seen when transaction taxes are imposed (see, for example, Wrobel \cite{wrobel1996financial}) may be less pronounced or even reversed in the modern era of high-frequency trading.

\section{A Model of a Limit Order Book market.}

\subsection{Basic trades}
We consider two types of orders: \emph{market} and \emph{limit} orders.

A limit order is an order to sell or buy a certain number of shares of an asset at a specified price, some time in the future. A trader submits his limit order to the exchange, where it is added to the \emph{limit order book}.

A market order is an order to immediately buy or sell a certain number of shares at the most favorable price available in the limit order book. A market order to buy  (to sell) is executed against the limit order to sell (to buy). The lowest specified price in the limit order book for a sell order is called the ask price, the highest price of a buy order in limit order book is called the best bid price.

We assume that at any time $t\le T$ traders can observe the limit order book against which market orders at time $t$ will be executed. We represent the limit order book by two functions $\rho^{+}_{t-}$ and $\rho^{-}_{t-}$ (we use a subscript $t-$ to indicate that this is the limit order book before any trades have occurred). Intuitively,
\[\rho^{+}_{t-}: \text{(price)} \to \text{(number of stocks)},\]
so $\rho^+_{t-}(s)$ tells us how many stocks can be purchased at the price $s$ at time $t$. Formally, we shall take $\rho^+_{t-}$ as the density of the sell orders on the limit order book, (against which market buy orders will be executed). Similarly, $\rho^{-}_{t-}$ represents the density of the buy orders on the limit order book (against which market sell orders will be executed). Naturally, $\rho^+_{t-}$ and $\rho^-_{t-}$ are nonnegative, however we shall not assume that they are strictly positive.

Therefore, as $\rho^+_{t-}$ and $\rho^-_{t-}$ are densities, we define
\[
F_{t-}^{+}(s) := \int_{0}^{s} \rho_{t-}^{+}(x) dx \quad \hbox{and} \quad F_{t-}^{-}(s) := \int_{s}^{\8} \rho_{t-}^{-}(x) dx
\]
so $F_{t-}^+(s)$ denotes the number of stocks offered for sale on the limit order book for a price less than $s$, while $F_{t-}^-(s)$ is the number of stocks offered to buy on the limit order book for a price at least $s$. Both these quantities determine the cost of trading using market orders.

Suppose a market order to buy $x$ stocks arrives. The right quasi-inverse function,
\[(F^+_{t-})^{-1_+}: \text{(number of stocks)}  \to \text{(new ask price)}\]
(where, $(F^+_{t-})^{-1_+}(x):= \inf\{y: F^+_{t-}(y)>x\}$) tells us the ask price after $x$ stocks have been bought using market orders, that is, the price needed to purchase one more stock (formally, an infinitesimal number more stocks) than $x$. For convenience, we denote this quantity
\begin{equation} \label{eq:D}
 D_{t-}^{+}(x) := (F_{t-}^{+})^{-1_+}(x).
\end{equation}
Similarly, the left quasi-inverse $(F^-_{t-})^{-1_-}(x) = \sup\{y:F^-_{t-}(y)<y\}$ tells us the bid price after $x$ stocks have been sold, and we write $D_{t-}^{-}(x) := (F_{t-}^{-})^{-1_-}(x)$.

Using this notation, we can naturally define the  ask $s^{*}$  and bid $s_{*}$ prices by
\[
 s^{*} := D_{t-}^+(0)=\inf_{\e>0} D_{t-}^{+}(\e) \quad \hbox{and} \quad s_{*} := D_{t-}^-(0)=\sup_{\e>0} D_{t-}^{-}(\e).
\]

\begin{lemma}
 The total cost to buy $x$ stocks using a market order is
\[
 H^{+}_{t-}(x) := \int_{s^*}^{D^{+}_{t-}(x)} u \rho^{+}_{t-}(u) du= \int_0^x D^+_{t-}(y) dy,
\]
similarly revenue from selling $x$ stocks using a market order is
\[
 H^{-}_{t-}(x) := \int_{D^{-}_{t-}(x)}^{s_*} u \rho^{-}_{t-}(u) du = \int_x^0 D^-_{t-}(y)dy.
\]
\end{lemma}
\begin{proof}
We prove only the equality of the final two terms in the first equation. As $F^+_{t-}$ is differentiable, its right quasi-inverse $D^+_{t-}$ is differentiable except at its discontinuities, that is, when $\rho^+_{t-}(s)=0$. These points do not contribute to the first integral, as $u\rho^+_{t-}(u)=0$ on such a set. Hence, we can exclude all such points from the first integral and use the classical change of variables result to obtain the desired equality.
\end{proof}

\begin{lemma}
We have the identities
\[\begin{split}
   \partial_{x+} H^+_{t-}(x) &= \lim_{h\downarrow 0} \frac{H^+_{t-}(x+h) - H^+_{t-}(x)}{h} = D^+_{t-}(x)\\
   \partial^2_{x+} H^+_{t-}(x) = \partial_{x+} D^+_{t-}(x) &= \lim_{h\downarrow 0} \frac{D^+_{t-}(x+h)-D^+_{t-}(x)}{h} = \frac{1}{\rho_{t-}^+(D_{t-}^+(x))}
  \end{split}
\]
If we assume that $\rho_{t-}^+$ also has an upper derivative, we also obtain
\[
 \partial^3_{x+} H^+_{t-}(x) = \partial^2_{x+} D^+_{t-}(x) = -\frac{\partial_{s+}\rho_{t-}^+(D_{t-}^+(x)) \cdot \partial_{x+} D_{t-}^+(x)}{(\rho_{t-}^+(D_{t-}^+(x)))^2}= -\frac{\partial_{s+}\rho_{t-}^+(D_{t-}^+(x))}{(\rho_{t-}^+(D_{t-}^+(x)))^3}.
\]
Similarly, with lower derivatives,
\[\partial_{x-} H_{t-}^-(x) = D^-(x), \qquad \partial_{x-}^2 H_{t-}^-(x) = \frac{-1}{\rho_{t-}^-(D_{t-}^-(x))},\qquad \partial_{x-}^3 H_{t-}^-(x) = \frac{\partial_{s-}\rho_{t-}^-(D_{t-}^-(x))}{(\rho_{t-}^-(D_{t-}^-(x)))^3}.\]
\end{lemma}
\begin{proof}
 These all follow from the definition of $H^+_{t-}$ and the chain rule (on the set $x:\rho^+_{t-}(D_{t-}^+(x))\neq 0$).
\end{proof}
\begin{remark}
 In particular, we note that
\[\partial_{x+} D^+_{t-}(0) = (\rho_{t-}^+(s^*))^{-1}.\]
\end{remark}

\begin{lemma}\label{lem:Taylorapprox}
If $\rho_{t-}^+$ has an upper derivative at $s^*$, then for $x\geq 0$,
\[\begin{split}
   H_{t-}^+(x)&=s^*x + \frac{1}{2\rho_{t-}^+(s^*)}x^2 -\frac{\partial_{s+}\rho_{t-}^+(s^*)}{6(\rho_{t-}^+(s^*))^3}  x^3 +o(x^4),\\
D_{t-}^+(x) &= s^* + \frac{1}{\rho_{t-}^+(s^*)}x - \frac{\partial_{s+}\rho_{t-}^+(s^*)}{2(\rho_{t-}^+(s^*))^3} x^2 + o(x^3),
\end{split}
\]
and similarly for $H^-_{t-}$ and $D^-_{t-}$. If we assume that the limit order book has an affine density above the ask price (that is, $\partial_{s+}\rho_{t-}^+(s)$ is constant for $s\geq s^*$), then the remainder terms vanish. If the limit order book has a constant density above the ask price (that is, $\partial_{s+}\rho_{t-}^+(s)\equiv 0$, as modelled in \cite{jarrow2010liquidity, roch2011resilient, roch2009liquidity}, and empirically studied in \cite{blais2010analysis}), then
\[
 H_{t-}^+(x) =s^*x + \frac{1}{2\rho_{t-}^+(s^*)}x^2,\qquad D_{t-}^+(x) = s^* + \frac{1}{\rho_{t-}^+(s^*)}x.
\]
For simplicity,  in this case we shall say that the limit order book has constant density.
\end{lemma}
\begin{proof}
 This is simply an application of Taylor's theorem.
\end{proof}

\begin{remark}
From these equations, we can clearly see the effect in $H^+_{t-}$ of the ask price (first order), the simple trading impact (second order) and the change in the limit order book height, or equivalently, the curvature of $D^+_{t-}$ (third order).
\end{remark}

\section{Fast trading and slippage}\label{sec:slippage}

From here onwards, we shall consider trades at a given time $t$. We shall therefore omit the $t-$ from $H^+_{t-}$, etc... whenever this does not lead to confusion.

 We wish to study the effects of Alice's priority to Bob, that is, the consequences of the fact that, when Alice and Bob both trade at the same time, Alice's order will be executed before Bob's. The first and simplest effect of Alice's priority is an increase in `slippage', where Bob's trade is executed further along the limit order book than if Alice was absent. This occurs when both Alice and Bob trade at the same time, in the same direction. For simplicity, we shall assume that they both try to purchase stock, the analysis if they both try to sell stock is perfectly analogous.

Let Alice attempt to buy a quantity $x\geq 0$, and Bob a quantity $y\geq 0$. Then Bob's order will not be executed at the front of the limit order book, but is affected by being executed after Alice's order. That is, Alice will pay the usual total cost $H^+(x)$, but Bob will pay the higher cost $H^+(x+y)- H^+(x)\geq H^+(y)$. Bob's loss due to slippage can then be measured by the difference between this quantity and $H^+(y)$, the amount Bob would usually pay.

\begin{lemma}\label{lem:slippagecost}
 The total cost of Bob's trade is given by
\begin{equation}\label{eq:slippageapprox}
H^+(x+y)- H^+_t(x)=  H^+(y) + \frac{1}{\rho^+(s^*)}xy -\frac{1}{2}\frac{\partial_{s+}\rho^+(s^*)}{(\rho^+(s^*))^3} \left( x^2y +xy^2\right) +o(x^3,y^3)
\end{equation}
If the limit order book has a constant density
\[ H^+(x+y)- H^+(x)=  H^+(y) + \frac{1}{\rho^+(s^*)}xy.\]
\end{lemma}
\begin{proof}
 Simply expand $H^+(x+y), H^+(x)$ and $H^+(y)$ using Lemma \ref{lem:Taylorapprox}.
\end{proof}

In the absence of Alice, Bob would expect to pay the quantity $H^+_t(y)$. From this approximation we can see that, provided the density of the limit order book at the ask price is sufficiently high, or the density is approximately constant at the ask price (so the third term disappears), for small trades, Bob's loss due to slippage is proportional to the covariation of Alice and Bob's trades (as measured by $xy$), with the proportion given by the inverse of the height of the limit order book density at the ask price.

\section{Fast Limit orders and Latency Arbitrage}
In the situation considered in Section \ref{sec:slippage}, Alice has a clear competitive advantage over Bob. However, she has not realised an arbitrage profit, as she has a entered into a net position in the stock. To allow pure arbitrage, where Alice starts and ends with no net position in the stock, but obtains a profit through trade, we also need to consider how Alice can place orders within the limit order book.

In Jarrow and Protter \cite{jarrow2011dysfunctional}, it has been shown that unequal access to to stock exchange may lead to some types of arbitrage. They have assumed that the fast trader invest according to an optional (rather than predictable) strategy, and realises the gains from the optional Ito integral (see section \ref{sec:optionalintegration}). This mathematical structure allows a high speed trader to obtain abnormal profits, as they can capture an additional term due to the quadratic variation of the stock.

In \cite{jarrow2011dysfunctional}, there is no significant discussion of how, in a real market, the mathematical formalism of an optional stochastic integral could be realised, or through what economic means these additional gains can be obtained.  In this section we will show that, by allowing Alice to also make fast limit order trades, if Alice foreknows Bob's trading strategy, these additional profits can be achieved in our model.

Suppose that Alice foreknows that Bob will purchase $y>0$ stocks. (Again, the analysis if Bob will sell stock is perfectly analogous.) Then Alice can make a profit with zero risk using the following recipe, where Alice front-runs Bob's trades. Recall that all of Alice's actions will be executed before Bob's trade.

\begin{str}\label{strategy1}
Suppose Alice knows that Bob will purchase $y$ stocks. Then Alice acts as follows:
\begin{enumerate}
 \item Purchase $x$ stocks using market orders. This will clear the limit order book up to the point $D^+(x)$.
 \item Place a limit order to sell $y\wedge x$ stock at the price $D^+_t(x)$ (or distributed infinitesimally below, so that these limit orders will be executed first). This is the limit order against which Bob will trade.
 \item Sell remaining $(x-y)^+$ stock using market orders.
\end{enumerate}
\end{str}

\begin{lemma}
 For general $x$, Alice's profit under this strategy is given by
\begin{equation}\label{eq:profit}
 \pi(x; y) = -H^+(x) + (y\wedge x) D^+(x) + H^-((x-y)^+).
\end{equation}
If $\pi(x;y)>0$, we say that Alice has realised a latency arbitrage opportunity. We call $\pi(x;y)$ the latency profit (or, from Bob's perspective, the latency cost).
\end{lemma}
\begin{proof}
Alice's purchase of $x$ stocks will cost her $H^+_t(x)$, and will clear the limit order book up to the price $D^+(x)$. Alice then places her limit order to sell $y\wedge x$ stocks for a price $D^+(x)$, so that these orders are the lowest in the order book. Bob's trade to purchase $y$ stock is executed, and as $y\geq y\wedge x$ all of Alice's limit order will be executed with Bob. This gives Alice revenue of $(y\wedge x) D^+(x)$. Finally, Alice sells any excess stock using market orders. As no previous trades have been executed against the lower side of the limit order book, the revenue from this is given by $H^-((x-y)^+)$.
\end{proof}

\begin{lemma}\label{lem:x=yoptimal}
If the limit order book has a constant density, Alice's profit is maximised by trading the volume $x=y$.
\end{lemma}
\begin{proof}
 As the limit order book has a constant density, $\rho^+(s)$ is independent of $s$ and $D^+(x)$ is linear, as in Lemma \ref{lem:Taylorapprox}. Hence, if the optimal value of $x$ is greater than $y$, by a first order condition,
 \[\begin{split}
    0 & = \frac{\partial}{\partial x}\pi(x,y)= -D^+_t(x) + y \frac{1}{\rho^+(D^+(x))} + D_t^-(x-y)\\
 y &= \rho^+(D^+(x))\left( D^+(x) - D^-(x-y)\right)\\
&= \rho^+(s^*) \left( s^* + \frac{x}{\rho^+(s^*)} - s_* + \frac{x}{\rho^-(s_*)}\right)\\
&= (s^*-s_*)\rho^+(s^*) + \left(1+ \frac{\rho^+(s^*)}{\rho^-(s_*)}\right)x\\
x &=\frac{ y - (s^*-s_*)\rho^+(s^*)}{ 1+ \frac{\rho^+(s^*)}{\rho^-(s_*)}} <y
   \end{split}
 \]
giving a contradiction. Therefore $x\leq y$. On the other hand, for $x \leq y$ we have
\[\begin{split}
   \frac{\partial}{\partial x}\pi(x,y)&= -D^+(x) + x \frac{1}{\rho^+(D^+(x))} + D^+(x)\\
&=\frac{x}{\rho^+(s^*)}>0
  \end{split}
\]
and so, by a first order condition, $x=y$ is optimal.
\end{proof}

\begin{remark}
 When the limit order book does not have a constant density, in particular if $\partial_{x+}\rho^+(s^*)<0$, then Alice's impact on the ask price $D^+(x)$ is increasing. As she is able to sell $y$ stock to Bob for $D^+(x)$, this yields higher revenues, and if these revenues are increasing sufficiently fast, they may compensate for the loss of having to sell excess purchased stock on the lower side of the limit order book.

 As discussed in \cite{gould2010limit}, newer exchange systems may require `market' orders to specify both a quantity and a maximum acceptable price (for buy orders, a minimum for sell orders), rather than simply a quantity. This forms a protection against extreme fluctuations in price, as were seen in the `flash crash' of May 6, 2010. If we were to model a market of this type, then the above result would be valid independently of the shape of the limit order book, as long as Bob specified that he wished to buy $y$ stocks for a maximum price of $D^+(y)$. For this reason, and for mathematical tractability, we shall hereafter assume that $x\leq y$ whenever Alice knows Bob's actions perfectly.
\end{remark}

\begin{lemma}\label{lem:aliceprofitapprox}
In the case $x\leq y$, we have the following equation for Alice's profit
\begin{equation}\label{eq:profitapprox}
\begin{split}
\pi(x) &= -H^+(x) + x D^+(x)\\
&= \frac{1}{2\rho^+(s^*)}x^2 - \frac{\partial_{s+}\rho^+(s^*)}{3(\rho^+(s^*))^3}  x^3 +o(x^4).
\end{split}
\end{equation}
As before, if the limit order book has an affine density, then this equation is exact, and if the limit order book has a constant density, the $x^3$ term can be omitted.
\end{lemma}
\begin{proof}
 Simply expand (\ref{eq:profit}) using Lemma \ref{lem:Taylorapprox}.
\end{proof}

\begin{lemma} \label{lem:LOB}
 If $x\leq y$, then Alice's trades have no impact on the shape of the limit order book following Bob's trades.
\end{lemma}
\begin{proof}
We will only prove the lemma for the ask-side of the limit order book, the bid-side can be proved by symmetry. For clarity, we denote all quantities before any trades have occurred with a subscript $t-$, those which are after Alice but before Bob with a subscript $t_A$, and those after both Alice and Bob with a subscript $t$.

Suppose that the Alice is not present at the market. Before any trades are executed, at time $t-$, the limit order book is described by the function
\[
 F_{t-}^{+}(s) = \int_{s^{*}_{t-}}^{s}  \rho^{+}_{t-} (u) du.
\]
Then Bob's market order is executed and a new ask is given by $s_{t+} = D_{t-}^+(y)$, and the ask-side limit order book has the form
\[
F_{t}^{+}(s) =  \int_{D_{t-}^{+}(y)}^{s}  \rho^{+}_{t-} (u) du.
\]
Now suppose that Alice is present. As Alice purchases $x$ stocks using market orders, then places a market order at $D_{t-}(x)$ to sell $x$ stocks, after Alice's trades, the ask-side limit order book has the form
\[
  F_{t_A}^{+}(s) = x \delta(D_{t-}^{+}(x)) +\int_{D_{t-}^{+}(x)}^{s}  \rho^{+}_{t-} (u) du,
\]
where $\delta(\cdot)$ denotes a Kronecker delta.

Then Bob's market order arrives, and buys $x$ stocks at the price $D_{t}^{+}(x)$ and $y-x$ stocks from the remainder of the limit order book. This implies that the net purchase from the limit order book is of precisely $y$ stocks. As a result, the ask-side limit order book after Alice and Bob's transactions has the form
\[F_{t}^{+}(s) =  \int_{D_{t-}^{+}(y)}^{s}  \rho^{+}_{t} (u) du,\]
the same as in the absence of Alice.
\end{proof}
The Lemma shows that in the case where $x\leq y$, the overall price dynamics are not changed due to the presence of Alice. We shall see (Section \ref{sec:volume}) that Alice's presence is still detectable, by considering the overall volume traded.

\subsection{Limits of arbitrage}

In the setting we have been considering, it is clear that Alice can obtain a risk-free profit, as soon as Bob trades. We can now answer two questions
\begin{enumerate}
 \item Where is Alice's profit coming from?
 \item What are the restrictions on Alice's profit?
\end{enumerate}

From an economic perspective, Alice's profit comes from an opportunity cost faced by Bob. Presuming $x=y$, if Bob was able to trade against the original limit order book, his cost of entering his position would be $H^+_t(y)$. Instead, he faces the cost $yD^+_y(y)$, which implies an opportunity cost of $\pi(y)$, Alice's profit. If Bob is trading because he believes he has determined a mispricing, then, presuming Bob is correct, Alice has acted as a classical arbitrageur, and Bob faces trading at an efficient market price. On the other hand, if Bob is a noise trader, then one can consider Bob as creating mispricing opportunities, which Alice is able to exploit.

In this situation, we note that Alice is making a risk-free profit. However, Bob does not face as large a price penalty as he would if Alice was simply copying Bob's strategy. This is easily seen by comparing the approximations (\ref{eq:slippageapprox}) and (\ref{eq:profitapprox}). In (\ref{eq:slippageapprox}), with $x=y$ we see that Bob's penalty is (to second order) given by $(\rho_t^+(s^*))^{-1}x^2$, whereas in (\ref{eq:profitapprox}) the second order term is halved.

In \cite{jarrow2011dysfunctional}, the fact Alice can obtain an arbitrage profit led Jarrow and Protter to conclude that an equivalent martingale measure may not exist, as there is a `Free Lunch with Vanishing Risk' in such a market, contradicting the fundamental theorem of Delbaen and Schachermayer \cite{delbaen2006mathematics}. This is clearly true in our setting. As the abnormal profit grows proportionally to $x^2$, the square of the amount traded by Alice, but they placed no bounds no Alice's strategies, in \cite{jarrow2011dysfunctional}, this may lead to Alice obtaining an unbounded profit.

In our approach, Alice is not able to make an unbounded profit. As Alice trades the amount $x\leq y$, and $y$ is chosen by Bob, not Alice, we see that Alice's profit cannot be scaled beyond this point. Therefore, in this simple model, even though Alice makes a bounded profit with zero risk, we still have no unbounded profit with bounded risk. This connects our analysis with weaker concepts of no-arbitrage, as considered by Karatzas and Kardaras \cite{karatzas2007numeraire}.

This situation is not economically unreasonable, as we have only considered whether Alice is capable of making an arbitrage profit on the market, without considering extra-market costs associated with doing so. In order to realise this profit, Alice needs to invest significant quantities in the development of fast trading systems, and in arranging for these systems to be collocated with the exchange servers. These costs form a significant barrier to entry in this market, preventing the arbitrage opportunity from being universally exploited. Furthermore, as Alice must continue to be the fastest trader in this section of the market, a form of the Red Queen Effect\footnote{From Lewis Carrol's \emph{Through the Looking Glass}, where the Red Queen states ``It takes all the running you can do, to keep in the same place.''} will force Alice to continually improve her systems, thereby continuing to incur such costs.

\subsection{Transaction costs and portfolio valuation}\label{sec:portfolioval}

Suppose that Alice is not purely attempting to make an arbitrage profit, but also wishes to take up a position in the stock. In this section, we shall see that Alice can use her speed advantage to shift her liquidity cost to Bob, whenever Bob's trades are of a similar size in the same direction as Alice's.

\begin{theorem}\label{thm:aliceportprofit}
Suppose Alice and Bob both wish to purchase stock at the same time. Alice's desired net trade is denoted by $x$, Bob's by $y$. For simplicity, assume $x,y\geq 0$ (as always, the analogous result holds if $x,y\leq 0$).

By exploiting her foreknowledge and increased speed, the cost to Alice of entering into this position can be reduced to
\[H^+(x)-yD^+(x+y) = s^{*} x+ \frac{1}{2\rho^{+}(s^{*})}   \left( 1 - \frac{ y^2}{ x^2} \right)x^{2} + o((x+y)^3).\]
This implies that, if $x\leq y$, then Alice will pay no more than the ask price $s^*$ for each unit of stock (ignoring the $o((x+y)^3)$ term).
Bob's total cost is given by
\[\begin{split}
   y D^+(x+y) &= y\left(s^* + \frac{1}{\rho^+(s^*)} (x+y) + o((x+y)^2)\right)\\
&=H^+(y) + \frac{1}{\rho^+(s^*)} yx +\frac{1}{2\rho^+(s^*)} y^2+ o((x+y)^2, y^3).
  \end{split}
\]
That is, Bob pays both the slippage cost $\frac{yx}{\rho^+(s^*)} $, as in Lemma \ref{lem:slippagecost}, and the latency arbitrage cost $\frac{y^2}{2\rho^+(s^*)} $, as in Lemma \ref{lem:aliceprofitapprox}.
\end{theorem}
\begin{proof}
As Alice wishes to enter into a position $x$, she can follow the following strategy.
\begin{str}
Alice uses a modified version of her earlier strategy.
 \begin{enumerate}
  \item Purchase a total of $x+y$ stocks using market orders, this will clear the limit order book up to $D^+(x+y)$.
  \item Place a limit order to sell $y$ stocks at the price $D^+(x+y)$. These limit orders are then executed against Bob's incoming market buy order.
 \end{enumerate}
\end{str}

Alice's total cost from these trades is given by $H^+(x+y)-yD^+(x+y)$, and she is left with a net position of $x$ stocks. Expanding $H^+$ and $D^+$ using Lemma \ref{lem:Taylorapprox}, we have
\[\begin{split}
 &  H^+(x+y)-yD^+(x+y)\\
&= \left(s^*(x+y) +\frac{1}{2\rho^+(s^*)} (x+y)^2 + o((x+y)^3)\right) + y\left( s^* + \frac{1}{\rho^+(s^*)} (x+y) + o((x+y)^2)\right)\\
&= s^*x +\frac{1}{\rho^+(s^*)}(x+y) \left(\frac{1}{2} (x+y) - y\right) + o((x+y)^3)\\
&= s^*x +\frac{1}{2\rho^+(s^*)} \left(1-\frac{y^2}{x^2}\right)x^2 + o((x+y)^3)\\
  \end{split}
\]
If Alice follows this strategy, Bob will pay the cost $yD^+(x+y)$. The expansion of Bob's cost again follows from Lemma \ref{lem:Taylorapprox}.

\end{proof}
In the case $x =  y$ this means that Alice can completely avoid the liquidity cost associated with her price impact. Furthermore, this leads to the strange situation where Bob and Alice hold the same portfolio, but due to differences in execution, they purchased it at a different cost.

It is a simple exercise to value a portfolio over given time interval in the same spirit as Roch \cite{roch2009liquidity}.
This will allow us to see that Alice and Bob will hold the same portfolios with same liquidation value, but Bob's liquidity costs due to price impact have quadrupled (from the $\frac{y^2}{2\rho^+(s^*)}$ term in $H^+(y)$ in Lemma \ref{lem:Taylorapprox} to the total cost of $\frac{2y^2}{\rho^+(s^*)}$ here), while Alice can trade as if she has no price impact whatsoever.

When Alice and Bob wish to trade in the opposite directions, then it is clear that they can both benefit from trading simultaneously. When their trades are of the same size, if Bob wishes to buy stock, Alice places the quantity she wishes to sell as a limit order at the front of the limit order book, and Bob trades against this. Bob benefits by reducing his price impact, and Alice benefits by being able to sell at the ask price, rather than at the bid (and by avoiding price impact).

\subsection{Optional integration}\label{sec:optionalintegration}
We now discuss how Alice's profits can be seen to come from an optional integral, more formally linking our approach with \cite{jarrow2011dysfunctional}.

 In \cite{jarrow2011dysfunctional}, our high-frequency trader Alice obtains a profit from being able to use an optional, rather than predictable, integral. This could be interpreted as a backwards integral in the sense of \cite{russo1993forward}, however in \cite{jarrow2011dysfunctional} the optional integral is only needed at a countable number of jump points, simplifying the mathematical analysis.

To be more precise Protter and Jarrow considered the following price process
\[ dS_{t} = S_{t-}\s(S_{t-})dZ_{t} + \eta dX_{t},\]
 where $Z$ and $X$  are semimartingales with respect to $\F_{t}$ such that $[X,Z] =0$ ($X$ and $Z$ have no common jumps).
$S_{t-}\s(S_{t-})dZ_{t}$ can be understood as  the fundamental value of the stock $S$ and $X$ is price impact of the
high-frequency trader (Alice).  The portfolio value for an ordinary trader who is using a predictable strategy $H$ is given by
\[ V_{H}(t) = H_{0} + \int_{0}^{t} H_{s-}S_{s-}\s(S_{s-})dZ_{s} + \int_{0}^{t}H_{s-}\eta dX_{s}.\]
Since Alice's strategy $X$ is assumed to be a predictable process except at its jump, her portfolio value is given by
\[
 V_{X}(t) = H_{0} + \int_{0}^{t} X_{s-}S_{s-}\s(S_{s-})dZ_{s} + \int_{0}^{t} X_{s-}\eta dX_{s} + \sum_{s\le t}\eta\D X_{s}^{2}.
\]
Hence the increased profit available to a fast trader (but unavailable to a slow trader) is
\[\eta\cdot(\sum_{s\le t}\D X_{s}^{2}) = \eta^{-1}\cdot(\sum_{s\le t:\Delta Z_s=0}\D S_{s}^{2}).\]

That is, considering only a single time $t$ where $\Delta X_t\neq 0$, if $\tilde\pi_{\text{JP}}$ denotes the change in excess profit at $t$, we have the relations
\begin{equation}\label{eq:jarrow}
 \Delta S_t = \eta \Delta X_t, \qquad \tilde\pi_{\text{JP}}=\eta^{-1}\cdot(\D S_{t})^{2}
\end{equation}

Now consider our model, as in Section \ref{sec:portfolioval}. Assuming a constant limit order book density above the ask price (equivalently, up to an approximation of appropriate order), we can calculate Alice's profit at a jump in the price. From Theorem \ref{thm:aliceportprofit}, Alice's increased profit at an upward jump in the price (that is, including her transaction costs, but not including the profits from a prior position) is given by
\[\tilde \pi(x,y) =x(S_{\text{book}})-\left(s^*x +\frac{1}{2\rho^{+}(s^{*})}   \left( 1 - \frac{ y^2}{ x^2} \right)x^{2}\right),\]
where $x$ and $y$ are the changes in Alice and Bob's positions respectively, and $S_{\text{book}}$ is the value at which Alice can value each stock she possesses. (Note that in \cite{jarrow2011dysfunctional} there is no bid-ask spread, so there is no ambiguity about the appropriate book value of a stock, that is, $S_{\text{book}}=S_{\text{ask}}=S_{\text{bid}}=S$.)

Suppose Alice's trade is of a constant size relative to Bob's trade, $x=\alpha y$ for some $\alpha>0$, so that
\[\tilde \pi = \alpha y (S_{\text{book}}-s^*) - \frac{\alpha^2-1}{2\rho^{+}(s^{*})} y^2.\]
The impact of this trade on the ask price is given by
\[(\Delta S_{\text{ask}}) = D^+_t(x+y) - s^* = \frac{1+\alpha}{\rho^+(s^*)}y\]
hence, when Alice and Bob both purchase stock, Alice's profit is
\[\tilde \pi = \alpha y(S_{\text{book}}-s^*) + \frac{\rho^+(s^*)}{2} \cdot \frac{1-\alpha}{1+\alpha} (\Delta S_{\text{ask}})^2.\]

Considering the other side of the limit order book, we obtain the general equation for Alice's profit,
\[\begin{split}
   \tilde \pi =& \left(\alpha y^+(S_{\text{book}}-s^*)+\frac{\rho^+(s^*)}{2} \cdot \frac{1-\alpha}{1+\alpha}  ((\Delta S_{\text{ask}})^+)^2\right)\\
& + \left(\alpha y^-(s_*-S_{\text{book}}) + \frac{\rho^-(s^*)}{2} \cdot \frac{1-\alpha}{1+\alpha}  ((\Delta S_{\text{bid}})^-)^2\right).
  \end{split}
\]

Now assume that  that $S_{\text{book}} = \frac{1}{2}(S_{\text{ask}}-S_{\text{bid}})$, the `mid-price', and that before trades occur, there is no bid-ask spread (a strong assumption of resilience in the market, but consistent with the lack of a Bid-Ask spread in \cite{jarrow2011dysfunctional}). Then when Alice and Bob both purchase stock,
\[\alpha y(S_{\text{book}} -s^*)= \frac{\alpha y}{2} \Delta S_{\text{ask}}=\frac{\alpha(1+\alpha)}{2\rho^+(s^*)}y^2 = \frac{\alpha}{1+\alpha} \cdot\frac{\rho^+(s^*)}{2}\cdot (\Delta S_{\text{ask}})^2 = \frac{\alpha}{1+\alpha} \cdot\frac{\rho^+(s^*)}{2}\cdot (2(\Delta S_{\text{book}}))^2 \]
and similarly when they sell stock. We can then write the profit as
\[\begin{split}
   \tilde \pi &= \frac{\rho^+(s^*)}{2} \cdot \frac{1}{1+\alpha}  (2(\Delta S_{\text{book}})^+)^2 + \frac{\rho^-(s^*)}{2} \cdot \frac{1}{1+\alpha}  (2(\Delta S_{\text{book}})^-)^2.
  \end{split}
\]
Assuming the symmetry
\[\rho^+(s^*) = \rho^-(s^*)= \frac{1+\alpha}{2}\cdot \eta^{-1}\]
for some price impact factor $\eta>0$, we finally have
\[\Delta S_{\text{book}} = \frac{\eta}{\alpha} x, \qquad \tilde\pi =\eta^{-1} \cdot (\Delta S_{\text{book}})^2\]
which is, for $\alpha=1$, precisely the abnormal benefit available to a high-frequency trader in (\ref{eq:jarrow}).

\section{Churning of trades and total volume}\label{sec:volume}

\begin{remark}
 From here to the end of this paper, we always assume that the limit order book has a constant density above the ask price and below the
bid price, as this allows us to give analytically simple values for the optimal trading levels for Alice and Bob. Therefore, for notational simplicity, we write simply write $\rho^{+}=\rho^{+}(s^*)$ and
$\rho^{-}_{t}=\rho^{-}(s^*)$.
\end{remark}

In the previous sections, we have described a possible strategy that will lead to Alice making a latency arbitrage profit. From Lemma \ref{lem:LOB} we see that although Alice momentarily trades on the limit order book, as her limit orders are immediately executed by Bob, the overall dynamics of the limit order book on any larger timescale are unchanged.

Therefore, from an econometric perspective, this type of arbitrage may not be noticed if we were to only analyse the shape of the limit order book through time. For that reason we introduce a new quantity $V_t^+$, the total volume traded using market buy orders at time $t$. The volume traded $V_t^+$ is the actual number of assets that have changed ownership through market buy orders at time $t$, and is readily observable on many markets.

As it will be important for us to keep track of whether we are referring to the shape of the limit order book before or after trades have occurred, we shall again subscript all quantities which refer to the limit order book before either Alice or Bob has traded with $t-$, all those which refer to after Alice's trades and before Bob's by $t_A$, and all those which are after both Alice and Bob's trades with $t$.

If each stock changed ownership at most once at time $t$, there would be a natural relationship between the ask price after a trade and $V^+_t$, given by (\ref{eq:D}). That is, if $V^+_t$ market buy orders have been executed, and no new limit orders have been placed on the market at the moment $t$, we would expect the ask price $s_{t}^*$ after all trades are executed to satisfy
\[s_{t}^* = D_{t-}(V^+_t)\]
or equivalently, $V^+_t = F_{t-}(s_t^*)$.

If Alice follows Strategy \ref{strategy1}, the volume traded through market buy orders will increase to $V_t^+ = x+ y$. However, as we have seen in Lemma \ref{lem:LOB}, if $x\leq y$ we know that $s_t^*$ remains unchanged. This breakdown of the relation between $V^+_t$ and the resultant ask price is due to Alice trading with both market and limit orders simultaneously. Using this observation, we can precisely define the degree of \emph{churning} (that is, of instantaneous round-trip transactions) present in the market.

\begin{defn}
Define the quantity
\[C_{t} := V_{t} - F_{t-}(s_t^*),\]
which measures the volume traded at time $t$ in excess of that which is implied by the change in the limit order book. We will say that the limit order book has been churned at time $t$ if $C_t>0$.
\end{defn}

We note that this definition does not work in the presence of cancellations of limit orders, as these may affect the ask price $s_t^*$ without requiring any transaction volume.

As $V_t^+$ is readily observable in many markets, as is the limit order book, the quantity $C_t$ provides a ready econometric quantity for the study of this type of high-frequency trading.

\subsection{Regulation and Tobin tax}\label{sec:simpletobin}

From a regulatory perspective, it may be advisable to discourage Alice's latency arbitrage trading, as it increases trade volumes (leading to higher administrative costs) and segregates markets according to access to high-frequency trading. A simple method of regulating this trading is for the market to charge Alice for access to high-frequency trading, either through an access charge, or through increased collocation costs. This is problematic, as the market (as a private corporation), then faces a conflict of interest, as they receive revenue from Alice's actions. This also does not assist Bob, as he will still only be able to trade at the higher price.

Another possible type of regulation, which is be suggested by the increased traded volumes, is to impose a transaction cost on all market participants, in the form of a financial transaction tax, or Tobin tax. This type of taxation is currently being proposed for certain transactions within the European Union (see \cite{europa}). We shall here propose a simple model of the tax, and study its affects on Alice's behaviour. As Alice is trading twice, it is natural to assume that she will face a larger tax burden than Bob, thereby discouraging her from attempting latency arbitrage.

In this section we simply consider Alice's behaviour in the presence of the tax, without looking for the effects of the tax on equilibrium behaviour of Bob, Alice and `market makers' (who provide the initial limit order book). The equilibrium analysis will be considered in Section \ref{sec:equilibriumTobin}.

\begin{defn}[Tobin tax]
 Suppose a market buy order is executed with a nominal monetary value $M$. The party placing a market order must pay tax on this transaction at a rate $r_m$, (that is, they must pay $(1+r_m)M$ total cost). The limit-order counterparty must pay tax on this transaction at a rate $r_l$, (that is, they only receive $(1-r_l) M$ payment).

We say the overall rate of taxation is
\[R:=\frac{1+r_m}{1-r_l}-1.\]
and that $r_m, r_l<1$ are chosen such that $R\in [0,1)$.
\end{defn}
\begin{remark}
 In some markets, there may be a premium paid to those placing limit orders, as a means of encouraging liquidity. This corresponds to $r_l<0$, and poses no problem in our setting, provided market orders are taxed at a rate $r_m>0$ such that $R\geq 0$.
\end{remark}

Let us analyse the impact of a Tobin tax on the market.
\begin{theorem}\label{thm:yminvalue}
In the presence of a Tobin tax with overall rate $R$, if Bob trades a quantity $y$ and
  \[ y \leq y_{\min} :=  2 \cdot \frac{R}{1-R}\cdot s^{*}\rho^{+} .\]
Alice's profit is maximised by not trading, that is, $x=0$. Otherwise, Alice's profit is maximised by trading as if the tax were not present, that is, $x=y$.
\end{theorem}
\begin{proof} 
That $x\leq y$ can be verified as in Lemma \ref{lem:x=yoptimal}.

If Alice makes a trade of size $x$, she faces the cost
\[
\left( 1+ r_{m} \right) H^{+} (x)
= \left( 1+ r_{m} \right) \left (s^{*} x +\frac{1}{2\rho^{+}} x^{2}  \right).
\]
Her trade changes the ask price to
\[ D^{+}(  x ) =  s^{*}+\frac{1}{\rho^{+}}  x .\]
She then submits a limit order of size $x$ (since $x\leq y$) that is executed against Bob's market order, giving her revenue
\[
\left( 1 - r_{l} \right) x D^{+}(  x ) 
= \left( 1 -  r_{l} \right) \left( \frac{1}{\rho^{+}} x^{2} + s^{*} x \right).
\]  
Hence her profit is
\[
 \begin{split}
  \pi(x) &= \left( 1 - r_{l} \right) x D^{+}(  x ) - \left( 1+ r_{m} \right) H^{+} (x)\\
&= \left( 1 -  r_{l} \right) \left( \frac{1}{\rho^{+}} x^{2} + s^{*} x \right) -  \left( 1+ r_{m} \right) \left (s^{*} x +\frac{1}{2\rho^{+}} x^{2}  \right)\\
&= -(r_m+r_l) s^* x + \frac{1}{\rho^{+}}\left(1-r_l - \frac{1+r_m}{2}\right) x^2
 \end{split}
\]
As this is a quadratic with positive coefficient of $x^2$, and $\pi(0)=0$, $\pi(x)$ is negative until the point $x=y_{\min}$, where
\[y_{\min} := \frac{(r_m+r_l) s^*}{\frac{1}{\rho^{+}}\left(1-r_l - \frac{1+r_m}{2}\right)} = 2 \cdot \frac{R}{1-R}\cdot s^*\rho^{+},\]
and is thereafter increasing and positive.

\end{proof}
Therefore, we see that a Tobin tax will only act to prevent Alice from exploiting small trades, and is completely ineffectual at preventing Alice from trading large quantities, particularly when the market is illiquid (i.e. $\rho^+$ is small).  Furthermore, we have the following corollary.
\begin{coro}
In the absence of Alice, the tax revenue from Bob's trade of $y$ will be
\[
 (r_m+r_l) H^+(y) = (r_m+ r_l) s^*y + \frac{1}{2\rho^{+}} \left(r_m+ r_l\right) y^2
\]
 If Alice is present and trades optimally, and $y>y_{\min}$, the tax revenue will be
\[(r_m+r_l) H^+(y) + (r_m+r_l) yD^+(y) = 2(r_m+ r_l)H^+(y) + \frac{r_m+r_l}{2\rho^{+}} y^2.\]
\end{coro}
From this, we see that whether the tax is placed on the market or limit order side of the transaction is irrelevant for the calculation of total revenue (under the assumption that Bob's trades and the original limit order book do not change, see Section \ref{sec:equilibriumTobin}). On the other hand, $R/(1-R)$ is more strongly affected by an increase in $r_l$ than $r_m$, so in terms of deterrence of Alice's trading, a tax on the limit order side of the transaction is more efficient. This effect is simply because, in our model, Alice trades a larger monetary value using limit orders than market orders ($xD^+(x)\geq H(x)$), and so responds more strongly to a tax on these orders.

\begin{remark}\label{rem:triggeringlimits}
 The fact that Alice will not profit from trading small quantities has implications for models when Alice can trigger Bob's trades. Suppose that Bob is an algorithmic trader, and that Alice knows Bob's algorithm. Suppose furthermore that  Alice can trigger Bob to trade a small quantity for negligible cost. Without a Tobin tax, Alice can use this fact to generate profits from Bob, by continuously triggering and front-running small trades. However, in the presence of a Tobin tax, Alice will only profit when Bob trades a quantity $y>y_{\min}$. So, if Alice is not able to trigger such trades without incurring significant costs herself, the introduction of a Tobin tax may have a stabilising effect on the market.
\end{remark}

\section{Imperfect knowledge}\label{sec:imperfectknowledge}

In all the preceding analysis, we have assumed that Alice knows perfectly the actions that Bob will take. Some reasons for this are given in Section \ref{sec:foreknowledge}.  When Alice does not perfectly know Bob's actions, but only has some prior distribution for the size and direction of Bob's trade, then her behaviour becomes more difficult to study. The main reason for this is that Alice must decide whether to take the risk of purchasing stock and placing it on the limit order book, not being sure whether Bob will purchase it. If Alice is unwilling to carry a position forward, any stock that Bob refuses to buy must be sold at a loss, using market orders executed against the bid side of the limit order book.

In an even more extreme case, if Alice is not sure of the direction of Bob's trade, and trades on the wrong side of the book, then the resolution of her position will occur after Bob's trades, that is, against the bid side of the limit order book with the first $x$ stocks removed. Alice then faces a slippage cost due to Bob, which will worsen her position further. Finally, it is conceivable that there are situations where Alice does not know the direction of Bob's trade, but the shape of the limit order book is such that it is in her interest to instantaneously modify both sides of the book. These considerations lead to a significantly more complex analysis.

In Moallemi, Park and Van Roy \cite{moallemi2008strategic}, a dynamic model for a predatory trader with Bayesian updating is considered. Under the assumption that all prior distributions are Gaussian and all strategies are linear in the prior's parameters, they derive a Nash equilibrium strategy for Alice and Bob. Here, we only consider Alice's behaviour, and only in a static setting. On the other hand, we do not assume that Alice's prior has any particular structure.

\subsection{Size uncertainty}
We now focus on one of the simplest forms of uncertainty, where Alice knows the limit order book perfectly, but is not able to perfectly predict Bob's actions. For simplicity, we suppose that Alice knows the direction, but not the size, of Bob's trade. We also assume that Alice is unable to carry a net position forward, and therefore will be forced to liquidate any excess purchased stock using a market order.

\begin{theorem}
Suppose that Alice assigns some subjective atomless probability $P$ to Bob trading an amount $y$. We assume that this probability already incorporates Alice's risk aversion, so that she simply wishes to maximise the expectation of her profit, $E_P[\pi(x,y)]$. Assume a constant density of the limit order book both above the ask price and below the bid price.

Then Alice's optimal trade is either $x=0$ or a solution to the nonlinear equation in terms of the lower partial moment $E_P[y|y<x]$ and the probability $P(y<x)$
\[x=\frac{\left(\frac{1}{\rho^{+}}+ \frac{1}{\rho^{-}}\right)E_P[y|y<x]-(s^*-s_*)}{\frac{2}{\rho^{+}}+\frac{1}{\rho^{-}} -\frac{1}{\rho^{+}P(y<x)}}.\]
\end{theorem}

\begin{proof}
As the limit order book density is constant, Alice's profit is given by
\[\begin{split} 
   \pi(x,y) &= -H^+(x) + (x\wedge y) D^+(x) - H^-((x-y)^+)\\
&= -s^*x - \frac{1}{2\rho^{+}}x^2 + (x\wedge y) (s^*x+ \frac{1}{\rho^{+}}x) +s_*(x-y)^+ - \frac{1}{2\rho^{-}}((x-y)^+)^2\\
&= \frac{1}{2\rho^{+}}(x\wedge y)^2-(s^*-s_*)(x-y)^+ - \frac{1}{2}\left(\frac{1}{\rho^{+}}+ \frac{1}{\rho^{-}}\right)((x-y)^+)^2.
  \end{split}
\]
As $P$ is atomless, the expectation of this quantity under $P$ is
\[ \begin{split}
    E_P[\pi(x,y)] &= \int \frac{1}{2\rho^{+}}(x\wedge y)^2-(s^*-s_*)(x-y)^+ - \frac{1}{2}\left(\frac{1}{\rho^{+}}+ \frac{1}{\rho^{-}}\right)((x-y)^+)^2 dP(y)\\
&=\int \frac{x^2}{2\rho^{+}}I_{x<y}-\left((s^*-s_*)x - \frac{1}{2}\left(\frac{1}{\rho^{+}}+ \frac{1}{\rho^{-}}\right)(x-y)^2\right)I_{x>y} dP(y)
   \end{split}
\]
For $x>0$, writing
\[p(x):=P(x>y)\quad \text{and}\quad q(x) := \int y I_{x>y} dP(y) = E_P[y|y<x] \cdot p(x),\]
 we have (formally using the lower derivative, so as to ensure that we have sufficient regularity to exchange the order of integration and differentiation)
\[ \begin{split}
\partial_{x-} E[\pi(x,y)] =& \int \frac{xI_{x<y}}{\rho^{+}}-(s^*-s_*)I_{x>y} - \left(\frac{1}{\rho^{+}}+ \frac{1}{\rho^{-}}\right)(x-y)I_{x>y} dP(y)\\
=& \frac{x}{\rho^{+}}(1-p(x)) - \left((s^*-s_*)+ \left(\frac{1}{\rho^{+}}+ \frac{1}{\rho^{-}}\right) x\right) p(x) \\
&+ \left(\frac{1}{\rho^{+}}+ \frac{1}{\rho^{-}}\right)q(x).
\end{split}
\]
Setting this derivative to zero, we obtain a nonlinear equation for the optimal quantity $x$ in terms of the lower partial moment $q(x)/p(x)=E_P[y|y<x]$ and the probability $p(x)$.
\[x=\frac{\left(\frac{1}{\rho^{+}}+ \frac{1}{\rho^{-}}\right)\frac{q(x)}{p(x)}-(s^*-s_*)}{\frac{2}{\rho^{+}}+\frac{1}{\rho^{-}} -\frac{1}{\rho^{+}p(x)}}.\]
\end{proof}

In many cases, this equation will not have a simple analytic solution, but is easy to solve numerically. The following example provides a case with a straightforward solution.
\begin{example}
Suppose Alice believes Bob's choice is uniformly distributed on $[0,\theta]$ for some $\theta>0$. Then for $x\in[0,\theta]$, $E[y|y<x] = x/2$ and $P(y<x)=x/\theta$. Therefore, Alice's optimal value of $x$ is
\[x = 2\cdot \frac{\theta - (s^*-s_*)\rho^{+}}{3+\frac{\rho^{+}}{\rho^{-}}}\wedge 0.\]

This simple equation agrees with our intuition. First, Alice will take a larger position when $\rho^{+}$ is small, as in this situation she can exploit Bob to a higher extent. Second she will not take any position if the spread $s^*-s_*$ is too large, as if Bob does not take up a sufficiently large position, she must sell at the lower price $s_*$. Conversely, when there is no spread, she will always take up a position. Finally, if $\rho^{-}$ is large, she will take up a larger position, as if she does need to sell excess stock, then she can do so without suffering from a significant price impact on the sale.

It is also interesting to note that, in this example, no matter what the shape of the limit order book, Alice will never take a position of more than $\frac{2}{3} \theta$. This is perhaps surprising, but is due to the fact that Alice's gain from purchasing the last portion of the limit order book is quite slight, and her potential loss from each stock purchased increases with the size of the trade. On the other hand, it is interesting to see that, for some shapes of the limit order book, Alice will purchase more than Bob's average trade $\frac{1}{2}\theta$, as her payoff is nonlinear in the size of the trade.
\end{example}

\subsection{Other uncertainty}
 There are also other possible sources of Alice's imperfect knowledge. In particular, we have assumed that simple market and limit orders are the only order types in the market, and that Alice knows the limit order book perfectly.

In recent years, new limit order structures, for example `iceberg orders', have been increasing in popularity. These orders explicitly attempt to hide the full shape of the limit order book, by only placing a small order on the book at any one time, but automatically replacing the order as soon as it is filled (see \cite{frey2009impact, esser2007navigation}). As this is done instantaneously, in particular before the remainder of Alice's order is executed, then Alice's impact on the price may be significantly less than she expects, with a consequent decrease in her profit. Furthermore, if Alice is unaware of this and does not correct the point at which she places her limit order to sell, then she may face significant losses, as Bob's trades will not be executed with Alice, leaving her with a net position in the stock, which must eventually be liquidated.

 These considerations require detailed analysis, in particular of the shape of the limit order book and of what knowledge Alice could have regarding those invisible orders present in the book.

\section{Equilibrium market efficiency}\label{sec:equilibrium}
We have previously mentioned that there is a sense in which Alice can be seen as an arbitrageur, if Bob is either trying to exploit mispricing or is a noise trader. In particular, if Bob detects a mispricing in the market (say, the ask price is below some level $p$). When Alice is present, he cannot trade at the mispriced level, rather, an attempt to purchase all mispriced stock be realised at the `fair' price $D^+(x)=p$. In this sense, we might think that Alice's presence results in a more efficient market.

On the other hand, when Bob becomes aware of Alice's presence in the market, he will adjust his behaviour accordingly. If Bob is an algorithmic trader, this corresponds to a modification of his algorithm, to adjust for the presence of Alice. Consequently, Bob may not fully exploit mispricing, and this will result in a less efficient market.

We can then determine the efficiency of the market by considering the economic surplus obtained in the market with and without Alice. For this analysis, we shall assume Alice faces no budget constraints or uncertainty, so that $x=y$. For simplicity, we shall yet again focus on the case of the purchase of stock, as the sale of stock is analogous.

We emphasise that this is not an equilibrium model for the provider of the initial limit order book, nor a temporal equilibrium model, but only a model for when Bob is aware of Alice.

\begin{notation}
 For clarity, we define the following notation. We shall denote by
\begin{itemize}
 \item $y^*$ the quantity of stock that Bob would trade if Alice were not present,
 \item $y_A$ the quantity of stock that Bob will trade if Alice is present.
\end{itemize}
And similarly for any other quantities of interest.
\end{notation}

\begin{defn}
 Let Bob have a marginal demand function denoted by $B(y)$, which denotes the price that Bob will pay to purchase another stock, given that he has already purchased $y$ stocks. We assume that $B$ is monotone weakly decreasing and continuous. If the total price of purchasing $y$ stocks is denoted  $\zeta(y)$, his economic surplus is given by
\[\gamma(y):= \int_0^y B(u) du- \zeta(y).\]
Hence if Alice is not present, Bob's economic surplus is given by
\[\gamma^*(y) := \int_0^{y} B(u)-D^+(u)du = \int_0^{y} B(u)du - H^+(y).\]
\end{defn}

\begin{lemma}
 In the absence of Alice, if the limit order book has a strictly positive density above the ask price, Bob maximises his economic surplus by trading an amount $y^*$, where either $y^*=0$ or $B(y^*)=D(y^*)$.  The value of $y^*$ is uniquely determined.
\end{lemma}
\begin{proof}
 This lemma simply states the first order condition that Bob will trade until the price offered on the market equals his marginal demand, that is, $B(y)=D^+(y)$. If this is never the case, then Bob will not trade (so $y^*=0$). As the limit order book has a strictly positive density, then $D^+(\cdot)$ is continuous and strictly increasing, so $y\mapsto B(y)-D^+(y)$ is continuous and strictly decreasing, and so will have a unique zero.
\end{proof}

Now suppose that Alice is present in the market. Assuming that Bob's marginal demand is linear, we can calculate the new size of Bob's trade, denoted $y_A$.
\begin{lemma}\label{lem:Bobreaction}
 Suppose that in the absence of Alice, $y^*>0$. As we know $B(y^*)=D(y^*)$, suppose that $B(y) = D(y^*) - b(y-y^*)$ for some $b\geq0$. Then in the presence of Alice, Bob's economic surplus is maximised by trading a quantity
\[y_A = \frac{1+b\rho^{+}}{2+b\rho^{+}} y^*>0.\]
\end{lemma}
\begin{proof}
 If Bob purchases $y$ stocks, in the presence of Alice his total cost is given by $yD^+(y)$, so his economic surplus is
\[\gamma_A(y):= \int_0^y B(u) du- yD^+(y).\]
The marginal price of each stock is given by
\[\begin{split}
  \frac{d}{dy} (yD^+(y)) &= D^+(y) + y\frac{d}{dy}D^+(y) = \left(D^+(y^*) + \frac{1}{\rho^{+}}(y-y^*)\right)+ \frac{1}{\rho^{+}}y\\
&= D^+(y^*) - \frac{1}{\rho^{+}}y^*+ \frac{2}{\rho^{+}}y.
  \end{split}
\]
And so from the first order condition,
\[ D^+(y^*) - b(y_A-y^*)= D^+(y^*) - \frac{1}{\rho^{+}}y^*+ \frac{2}{\rho^{+}}y_A,\]
and rearrangement yields the result.
\end{proof}
Note that, under this model, if Bob will trade in the absence of Alice, he will still trade in the presence of Alice. This gives an explanation of why Alice's profits could persist through time, as Bob will continue to trade even when fully aware of Alice's actions.

In extreme cases, when Bob purchases only up to a fixed price (i.e. $b=0$), or when Alice's price impact is large ($\rho^{+}\to0$), then Bob will only purchase half of what he would purchase if Alice were not present. On the other hand, if Bob is insensitive to the price or the limit order book is deep ($b\rho^{+}\to\infty$), Bob's behaviour will not be affected by Alice's presence.

A simple way to measure the cost to market efficiency which is introduced by Alice is by considering the deadweight loss. This is given by the loss in `economic surplus' (the integral of the difference between the price a party pays and what they are willing to pay) summed over all parties in the market. In the context of our model, we can calculate this quantity precisely.

\begin{theorem}\label{thm:deadweight}
 In the situation of Lemma \ref{lem:Bobreaction}, when Alice is present, Bob's economic surplus is reduced to
\[\gamma_A(y_A) := \int_0^{y_A} B(u) - \zeta(u) du = \int_0^{y_A} B(u)du - y_AD^+(y_A)\]
Alice's economic surplus is given by her profit,
\[\pi(y_A)= y_AD^+(y_A) - H^+(y_A)= \frac{1}{2\rho^{+}}(y_A)^2,\]
 and so the \emph{deadweight loss} introduced by Alice is given by
\[\begin{split}
\gamma^*(y^*)-\gamma_A(y_A) - \pi(y_A) &= \gamma^*(y^*)-\gamma^*(y_A) = \int_{y_A}^{y^*} B(u)-D^+(u) du\\
   &=\left(b+\frac{1}{\rho^{+}} \right)\left[\left(\frac{1+b\rho^{+}}{2+b\rho^{+}}-\frac{1}{2}\right)^2+\frac{1}{4}\right] (y^*)^2\\
&\geq \frac{1}{4\rho^{+}} (y^*)^2.
  \end{split}
\]
\end{theorem}
\begin{proof}
 We have $y_A = \frac{1+b\rho^{+}}{2+b\rho^{+}} y^*>0$. Hence, using the same expansions for $B$ and $D^+$ as in Lemma \ref{lem:Bobreaction},
\[\begin{split}
   \int_{y_A}^{y^*} B(u)-D^+(u) du &= \int_{y_A}^{y^*} -\left(b+\frac{1}{\rho^{+}}\right) (u-y^*) du\\
&=\left(b+\frac{1}{\rho^{+}} \right)\left[\left(\frac{1+b\rho^{+}}{2+b\rho^{+}}-\frac{1}{2}\right)^2+\frac{1}{4}\right] (y^*)^2.
  \end{split}
\]
Setting $b=0$ yields the lower bound.
\end{proof}

\subsection{Bob as ideal mispricing investor}
In a similar vein, we have the following economic interpretation. Suppose Bob is an ideal investor who notices a fundamental mispricing in the stock. Bob believes the fair price is $p>s^*$, and is willing to purchase any quantity of stock with a marginal price below $p$. As $p$ is the same in the presence or absence of Alice, we have $p=D^+(y^*)$, where $y^*$ is the amount of stock Bob buys when Alice is not present. We easily see that Bob's marginal demand function is constant ($b=0$), and we obtain both the lowest amount of trading ($y_A=y^*/2$) and the lowest deadweight loss, by Lemma \ref{lem:Bobreaction} and Theorem \ref{thm:deadweight}. (The fact that the Deadweight loss is minimal at the lowest amount of trading is surprising, but is due to the fact that Bob's surplus depends less on the price paid then when $b\neq 0$.)

From an economic perspective, we can see why Bob acts in this way. In the presence of Alice, if Bob attempts to purchase the entirety of the mispriced stock, then due to Alice's intervention he will pay the price $D^+(y^*)$ for each stock. Therefore, he will be unable to realise any profits due to mispricing.

On the other hand, as we have seen, Bob prefers to purchase only a fraction $y_A=y^*/2$ of the mispriced stock, for a price $D^+(y_A)<D^+(y^*)$, which will allow him to realise some profits. As a consequence, the market faces a deadweight loss in terms of Pareto efficiency, as the mispricing is not fully exploited. It is interesting to note that the deadweight loss is equal to $\frac{1}{\rho^{+}} (y_A)^2$, which is precisely double Alice's profit in this setting. In this sense, we can see that Alice's actions introduce a significantly larger economic cost than the profits she obtains.

Furthermore, rather than Bob immediately rectifying the mispricing by purchasing all undervalued stock at the first opportunity, he only moves the price partway towards the efficient price. If we model multiple trading periods, then Bob will eventually move the price to the efficient price, exploiting the remaining mispricing with each subsequent trade. However, one can see that this can induces the following pathological consequence:
\begin{quotation}
\emph{The presence of the high speed trader Alice, when Bob is a rational agent, can lead to Bob being unwilling to instantly and fully exploit fundamental mispricings in the market. Hence the market efficiency through time may be decreased by the presence of the high speed trader.}
\end{quotation}

Another interesting fact is that the total volume traded at equilibrium is independent of whether Alice is present or not, at least in the case $b=0$. If Alice is not present, a total of $y^*$ stocks are traded. If Alice is present, only $y_A=y^*/2$ stocks are traded, but every one of those stock is traded twice. Therefore, in either case, the total traded volume is $y^*$.

\subsection{Equilibrium with a Tobin tax}\label{sec:equilibriumTobin}

Finally, we consider the setting where a Tobin tax is introduced on all participants in the market. We have already seen that the effect of such a tax on Alice is only to reduce her profitability, and to prevent her exploiting Bob's small trades. We now wish to examine the consequences of such a tax on the deadweight loss that Alice and the tax collectively introduce to the market.

Suppose, as in Section \ref{sec:simpletobin}, that tax at a rate $r_l$ is paid by the party executing a limit order, and a rate $r_m$ by the party executing a market order. 

We consider the effect of this on three parties: Alice, Bob and the `market maker' (who provides the initial limit order book). As is usual, we shall consider only the situation where Bob wishes to purchase stock, and Alice has perfect foreknowledge of Bob's actions. We assume that the addition of the tax does not change Bob's preferences (so $B(y)$ is unaffected).

\begin{notation}
 In addition to what we had earlier, we now write
\begin{itemize}
 \item $y_{A,T}$ for the amount of stock Bob will trade in the presence of both Alice and the tax
\end{itemize}
and similarly for Bob's economic surplus $\gamma$. Recall that $y_{\min}$ is the minimum size of Bob's trade at which Alice will begin to act (c.f. Theorem \ref{thm:yminvalue}), and that $R=\frac{1+r_m}{1-r_l}-1$ is the overall rate of taxation, which we assume satisfies $R\in[0,1)$.
\end{notation}

We can now determine the behaviour of each participant in this market.

\begin{itemize}
\item In the absence of the tax, the market maker is willing to sell $y$ stocks at a total cost of $H^+(y)$. In the presence of the tax, the market maker will require compensation for the increased cost, and so they will be willing to sell $y$ stock for the increased price $\frac{1}{1-r_l} H^+(y)$.
\item If Bob elects to purchase $y$ stock, then Alice may be able to make a profit by purchasing $x$ stock for a pre-tax cost $\frac{1}{1-r_l} H^+(x)$, and then selling it to Bob for the pre-tax amount $\frac{1}{1-r_l} xD^+(x)$. Net of tax, Alice's profit is given by 
\[\pi(x, y) = xD^+(x) - \frac{1+r_m}{1-r_l} H^+(x)= xD^+(x) - (1+R) H^+(x).\]
Note that Alice's profit function here is the same as that in Section \ref{sec:simpletobin}, divided by $1-r_l$. Hence, by the same argument as in Section \ref{sec:simpletobin}, if we assume that the limit order book has a constant density above the ask price, Alice's profit is maximised by selecting $x=y$ whenever $y> y_{\min}$, and $x=0$ otherwise, when $y_{\min}$ is given by 
\[y_{\min}=2 \cdot \frac{R}{1-R}\cdot s^{*}\rho^{+}.\]
We assume that Alice will not trade at the point $y=y_{\min}$.
\item Bob's behaviour is now somewhat more complicated, due to the discontinuity in Alice's strategy.  Bob's cost of purchasing $y$ stock, net tax, is given by
\[\begin{split}
\text{Cost} &= \begin{cases} \frac{1+r_m}{1-r_l} H^+(y) = (1+R) H^+(y) & y\leq y_{\min}\\ \frac{1+r_m}{1-r_l}yD^+(y) = (1+R)yD^+(y) &  y> y_{\min}.\end{cases}\\
\end{split}\]
Assuming a constant density limit order book above the ask price, and linearity of Bob's marginal demand, Bob's economic surplus is given by
\[\begin{split}
\gamma_{A,T}(y) &= \int_0^y B(u)du - (1+R)H^+(y)- \frac{1+R}{2\rho^{+}} y^2 I_{y>y_{\min}}\\
&=yD^+(y^*)-b\left(\frac{y^2}{2}-y^*y\right) - (1+R)\left(s^*y + \frac{1}{2\rho^{+}} y^2\right) - \frac{1+R}{2\rho^{+}} y^2 I_{y>y_{\min}}\\
&= \left(-Rs^*+\left(b+\frac{1}{\rho^{+}}\right)y^*\right)y - \left(b+\frac{1+R}{\rho^{+}}\left(1+I_{y>y_{\min}}\right)\right) \frac{y^2}{2}
\end{split}
\]
\end{itemize}

\begin{theorem}\label{thm:equilibriumBobbound}
Assuming linearity of Bob's marginal demand, if 
\[y^*\geq 4\frac{R}{1-R}\cdot \frac{bs^*\rho^{+} + (2+3R-R^2)s^*}{b+\frac{1}{\rho^{+}}}\]
then Bob's economic surplus is maximised when he purchases the amount
\[y_{A,T}= \frac{\left(-Rs^*+\left(b+\frac{1}{\rho^{+}}\right)y^*\right)}{\left(b+2\frac{1+R}{\rho^{+}}\right)}>y_{\min}\]
where $y^*$ is the quantity that Bob would buy in the absence of \emph{both} Alice and the tax. In this situation, Alice will also trade the amount $y_{A,T}$.
\end{theorem}
\begin{proof}
As $\gamma$ is a quadratic with negative leading coefficient and a discontinuity at $y=y_{\min}$, we have four possible maxima.
\begin{itemize}
\item If $y=0$ then $\gamma_A(y) = 0$.
\item If $y=y_{\min}=2 \cdot \frac{R}{1-R}\cdot s^{*}\rho^{+}$ then 
\[\gamma_A(y)=\left(2 \cdot \frac{R}{1-R}\cdot s^{*}\rho^{+}\right)\left(-Rs^*+\left(b+\frac{1}{\rho^{+}}\right)y^* - \left(b+\frac{1+R}{\rho^{+}}\right) \left(\frac{R}{1-R}\cdot s^{*}\rho^{+}\right)\right)\]
\item If 
\[y=y_1:= \frac{\left(-Rs^*+\left(b+\frac{1}{\rho^{+}}\right)y^*\right)}{\left(b+\frac{1+R}{\rho^{+}}\right)}=y^*-\frac{Rs^*}{\left(b+\frac{1+R}{\rho^{+}}\right)}\leq y_{\min}\quad \text{then} \quad \gamma_A(y)=\frac{\left(-Rs^*+\left(b+\frac{1}{\rho^{+}}\right)y^*\right)^2}{2\left(b+\frac{1+R}{\rho^{+}}\right)}\]
\item If 
\[y=y_2:= \frac{\left(-Rs^*+\left(b+\frac{1}{\rho^{+}}\right)y^*\right)}{\left(b+2\frac{1+R}{\rho^{+}}\right)}>y_{\min}\quad \text{then} \quad \gamma_A(y)=\frac{\left(-Rs^*+\left(b+\frac{1}{\rho^{+}}\right)y^*\right)^2}{2\left(b+2\frac{1+R}{\rho^{+}}\right)}\]
\end{itemize}
It is clear that if $0<y_1\leq y_{\min}$, then $y=y_1$ is the optimal solution. If $y_1<0$, then $y=0$ is the optimal solution. Otherwise, we must decide between $y_{\min}$ and $y_2$, as a simple argument shows that if $y_2$ is preferred to $y_{\min}$, it is also preferred to $y_1$. If $y_2\leq y_{\min}$, then it is clear that $y_{\min}$ is preferable. Otherwise, there is no simple  condition which determines this, so we seek a useful sufficient condition under which $y_2$ is preferable to $y_{\min}$.

Suppose $2 y_{\min} \leq  y_2$. Then 
\[\begin{split}
\gamma_{A,T}(y_{\min}) &= y_{\min} \left(-Rs^*+\left(b+\frac{1}{\rho^{+}}\right)y^* - \left(b+\frac{1+R}{\rho^{+}}\right) \left(\frac{R}{1-R}\cdot s^{*}\rho^{+}\right)\right)\\
&\leq y_{\min} \left(-Rs^*+\left(b+\frac{1}{\rho^{+}}\right)y^*\right)\\
&= y_{\min}\cdot y_2 \left(b+2\frac{1+R}{\rho^{+}}\right) \\
&\leq \frac{(y_2)^2}{2} \left(b+2\frac{1+R}{\rho^{+}}\right)\\
&=\gamma_{A,T}(y_2)
\end{split}\]
and so if $2y_{\min}\leq y_2$, we know that $y_2$ is preferable to $y_{\min}$.

Finally, by expanding $y_{\min}$ and $y_2$ we can reduce the condition $2y_{\min}\leq y_2$ to
\[y^*\geq 4\frac{R}{1-R}\cdot \frac{bs^*\rho^{+} + (2+3R-R^2)s^*}{b+\frac{1}{\rho^{+}}}.\]
\end{proof}

 When $b=0$, we can simplify the inequalities in Theorem \ref{thm:equilibriumBobbound} further.
\begin{coro}
If $b=0$, then Theorem \ref{thm:equilibriumBobbound} reduces to:  If 
\[y^*\geq 4\frac{R(2+3R-R^2)}{1-R}\cdot s^*\rho^{+}= 2y_{\min}+o(R^2)\]
then Bob's economic surplus is maximised when he purchases the amount
\[y_{A,T}= \frac{1}{1+R}\cdot \frac{y^*}{2}-\frac{1}{2}\cdot \frac{R}{1+R}s^*\rho^{+}>y_{\min}\]
\end{coro}

Suppose the proposed overall tax rate $R$ is small and the market moderately illiquid (so $\rho^+$ is not large), and so $y_{\min}=2Rs^*\rho^{+} +o(R^2)$ is also small. Even when Bob is aware of Alice's presence, we see from this theorem that  Alice will only be prevented from exploiting those of Bob's trades which are inevitably small. Hence, the introduction of a Tobin tax at small levels is still generally ineffective at preventing Alice from exploiting large scale latency arbitrage. 

\begin{theorem}
 The introduction of a Tobin tax will improve Bob's economic surplus if and only if he trades a quantity $y$ where
\[0<y< 2\frac{R}{2-R}s^*\rho^{+} < y_{\min}\]
\end{theorem}
\begin{proof}
 Clearly the surplus of Alice and the market maker is reduced by the introduction of a Tobin tax, as they both suffer from Bob electing to trade a smaller quantity, and from the taxation which is levied on them. At first glance, one might think that Bob could profit from the introduction of a tax if he was planning on trading a quantity $y\leq y_{\min}$, as the presence of the tax will prevent Alice from participating, and this benefit to Bob could outweigh the costs to Bob of taxation.

To see if this is the case, we compare Bob's economic surplus when only Alice is present ($\gamma_A$), with that when both Alice and the Tax are present, but Alice may choose not to trade ($\gamma_{A,T}$).
 \[\begin{split}
    \gamma_A(y)-\gamma_{A,T}(y) &= \left(\int_0^y B(u)du -yD^+(y) \right) - \left(\int_0^y B(u)du - \begin{cases} (1+R) H^+(y) & y\leq y_{\min}\\ (1+R)yD^+(y) &  y> y_{\min}.\end{cases} \right)\\
 &= \begin{cases} R H^+(y) -\frac{1}{\rho^{+}} y^2 & y\leq y_{\min}\\ -R yD^+(y) &  y> y_{\min}.\end{cases}
   \end{split}
 \]
Clearly Bob is made worse off ($\gamma_A(y)>\gamma_{A,T}(y)$) for every $y>y_{\min}$. For $y\leq y_{\min}$,
\[ \gamma_A(y)-\gamma_{A,T}(y)= R H^+(y) -\frac{1}{\rho^{+}} y^2 =  (2Rs^*\rho^{+}-(2-R)y)\frac{y}{2\rho^{+}}\]
which is positive if and only if
\[y< 2\frac{R}{2-R}s^*\rho^{+}< y_{\min}.\]
\end{proof}

\begin{coro}\label{cor:boblikestobin}
 With optimal trading, Bob's economic surplus is increased by the introduction of the Tobin tax if and only if
\[\frac{Rs^* \rho^{+}}{1+R+ b \rho^{+}}< y^*<\left(\frac{2}{2-R} +\frac{1}{b\rho^{+} +1 + R}\right) Rs^* \rho^{+},\]
when he trades a quantity
\[y_{A,T} = y^* - \frac{Rs^* \rho^{+}}{1+R+ b \rho^{+}}>0\]
where $y^*$ is the amount Bob would trade in the absence of both Alice and the tax.
\end{coro}
\begin{proof}
 We know Bob's economic surplus is increased if and only if he trades a quantity
\[0<y< 2\frac{R}{2-R}s^*\rho^{+}< y_{\min}.\]
For this to be optimal, from the proof of Theorem \ref{thm:equilibriumBobbound} we must be in the case where
\[y= y_1 = y^* - \frac{Rs^* \rho^{+}}{1+R+ b \rho^{+}}<2\frac{R}{2-R}s^*\rho^{+}y_{\min}\]
and rearrangement yields the result.
\end{proof}

\begin{remark}
 This result provides a partial response to the suggestion that a Tobin tax will drive all transactions out of a market. If Bob can only trade in markets where Alice (or another high-frequency trader) also trades, then the imposition of a Tobin tax, by eliminating this predatory trading, can make a market \emph{more} attractive as an alternative, at least for trades of some sizes. This effect of the recent development of high-frequency trading could lessen the `flight of capital' historically observed in markets that introduce Tobin taxes (see \cite{wrobel1996financial}).
\end{remark}

Our final theorem gives conditions under which the market as a whole could be benefited by the introduction of a Tobin tax.
\begin{theorem}\label{thm:tobinisgood}
 Assuming Bob has linear preferences, the introduction of a Tobin tax results in a lower deadweight loss (relative to the market without the tax but with Alice) if and only if
\[\begin{split}
  (2+b\rho^{+})\frac{Rs^*\rho^{+}}{1+R+b\rho^{+}} < y^* < 2\frac{2+b\rho^{+}}{1+b\rho^{+}}\cdot \frac{R}{1-R} s^*\rho^{+}
  \end{split}
\]
where $y^*$ is the number of stocks Bob will trade in the absence of Alice and the tax.
In the case $b=0$, this implies that the Tobin tax is effective at reducing deadweight loss only when
\[  \frac{R}{1+R} < \frac{y^*}{2s^*\rho^{+}}= \frac{y_A}{s^*\rho^{+}} < \frac{2R}{1-R} .\]
\end{theorem}
\begin{proof}
For $b$ fixed, the deadweight loss is related in a monotone way to the number of stocks traded (more stocks traded equals less deadweight loss). Therefore, it is sufficient for us to consider under what circumstances the number of stocks that Bob choses to trade will increase.

If $y_A\geq y_{\min}$, then Bob will decrease the number of stocks he trades when the tax is introduced. This is because Alice will continue to trade against him if he choses to trade $y_A$, but he will face the increased costs associated with the tax. Hence if $y_A\geq y_{\min}$, we know $y_{A,T}<y_A$, and the deadweight loss has increased.

If $y_A< y_{\min}$ then Bob may chose to increase or decrease his trading, depending on the relative impact of Alice's trading and the tax. Clearly, he will not increase the number of stocks he trades beyond $y_{\min}$, therefore, from the proof of Theorem \ref{thm:equilibriumBobbound}, we have
\[y_{A,T} = y_1\wedge y_{\min} = \left(y^* - \frac{Rs^*\rho^{+}}{1+R+b\rho^{+}}\right)\wedge y_{\min}.\]
It follows that $y_{A,T}> y_A$ if and only if $y_1>y_A$ and $y_{\min}>y_A$. Writing out these two inequalities in terms of $y^*$ we have
\[\begin{split}
 \left(y^* - \frac{Rs^*\rho^{+}}{1+R+b\rho^{+}}\right) &> \frac{1+b\rho^{+}}{2+b\rho^{+}}y^*\\
 2\frac{R}{1-R} s^*\rho^{+} &> \frac{1+b\rho^{+}}{2+b\rho^{+}}y^*\\
  \end{split}
\]
and rearrangement yields the result.
\end{proof}
This theorem shows us that, in general, there is a range of transaction sizes such that the introduction of a Tobin tax is sufficient to prevent Alice from trading, but is not so large that the cost of the tax outweighs the benefit of Alice being excluded from the market. 

We also note that this range is relatively larger in liquid markets (where $\rho^+$ is large) than in illiquid markets (where $\rho^+$ is small). This is because Alice's profits in illiquid markets are much higher, and so the tax is less likely to prevent Alice from trading.

If the tax is only implemented on market orders (and even, perhaps, some rebates are given on making limit orders) the market maker will benefit from the tax, as their trade volumes will increase, without them paying any extra costs. Hence, if both the conditions of Corollary \ref{cor:boblikestobin} and Theorem \ref{thm:tobinisgood} are satisfied, Alice is the only party who suffers from the introduction of the tax (as the government, market makers and Bob all benefit).

\begin{remark}
 We see that if $b=0$, then both Corollary \ref{cor:boblikestobin} and Theorem \ref{thm:tobinisgood} will be satisfied when (approximately),
\[R+o(R^2)<\frac{y^*}{s^*\rho^+}= \frac{y_A}{2s^*\rho^+}< 2R+o(R^2).\]
Given that in a market with Alice the joint distribution of $y_A$, $s^*$ and $\rho^+$ is observable, this may help suggest values for $R$ for which slow traders can be enticed to trade on the market, and market efficiency can be improved.
\end{remark}

\section{Conclusions}

We have presented a single-moment model for latency arbitrage between two traders in the presence of a limit order book, where one trader has a significant speed advantage over the other. We have seen that in this model, a high speed trader with perfect foreknowledge of the actions of the slow trader can reap a positive profit with zero risk, by manipulating the limit order book. This profit is, however, bounded, as the high speed trader cannot force the slow trader to act. Furthermore, when the high speed trader has only a prior distribution for the slow trader's actions, a profit is still possible, however with some risk.

We have seen that without considering the slow trader's response, the total volume of trading in this setting increases, however without changing the overall dynamics of the limit order book. We have investigated the behaviour when the slow trader is aware of the actions of the fast trader, and shown that this leads to an overall reduction in market efficiency. By considering the introduction of a tax on financial transactions, we see that such a tax will either prevent the high speed trader from participating, or will have no impact on their behaviour. We have also given a range of tax rates when the market with both the tax and the high speed trader is more efficient than with only the high speed trader.

\bibliographystyle{plain}
\bibliography{LOB}

\end{document}